\documentclass[10pt,twocolumn,twoside]{IEEEtran}

\renewcommand{\baselinestretch}{0.98}

\usepackage{multicol, multirow}
\usepackage{array}
\usepackage{amsmath, amsthm, amsfonts, amssymb}
\usepackage{cite}
\usepackage{graphicx}
\usepackage{subfigure}
\usepackage{color}
\usepackage{cases}
\usepackage[]{footmisc}
\usepackage{calc}
\usepackage{physics}
\usepackage{booktabs}
\usepackage{placeins}
\usepackage{xcolor}
\usepackage{hhline}
\usepackage{kotex}
\usepackage{soul}
\usepackage{dblfloatfix}
\usepackage{algorithm,algorithmic}
\usepackage[normalem]{ulem} 

\newcommand{\Conv}{%
  \mathop{\scalebox{1.5}{\raisebox{-0.2ex}{$\circledast$}}
  }
}

\ifCLASSOPTIONonecolumn
   \renewcommand{\baselinestretch}{2.0}  
   \newcommand{\figwidth}{0.6\columnwidth}
   
   \newcommand{\widefigwidth}{0.5\columnwidth} 
 \else
   \newcommand{\figwidth}{0.9\columnwidth}  
   
   \newcommand{\widefigwidth}{0.95\columnwidth}   
\fi

\newcommand{\ignoreIt}[1] {}

\allowdisplaybreaks 
\newtheorem{theorem}{Theorem}
\newtheorem{corollary}{Corollary}
\newtheorem{lemma}{Lemma}
\newtheorem{prop}{Proposition}
\newtheorem{remark}{Remark}

\DeclareMathOperator*{\argmax}{argmax}

\begin{document}
\author{Jinyoung~Lee, Hyeonsik~Yeom \\
\thanks{
J. Lee is with the Division of Electronics and Electrical Information Engineering, National Korea Maritime \& Ocean University, Busan 49112, South Korea, email: jylee120@kmou.ac.kr.
    
H. Yeom is with the Defense ICT Convergence Research Section, Electronics and Telecommunications Research Institute (ETRI), Daejeon 34129, South Korea, email: hyeonsik@etri.re.kr.}
}

\title{Multi-User Covert Communication in Spatially Heterogeneous Wireless Networks}
\ifCLASSOPTIONonecolumn
  \renewcommand{\baselinestretch}{1.5}  
\fi

\maketitle

\begin{abstract}
This paper investigates an uplink multi-user covert communication system with spatially distributed users. Unlike prior works that approximate channel statistics using averaged parameters and homogeneous assumptions, this study explicitly models each user’s geometric position and corresponding user-to-Willie and user-to-Bob channel variances. This approach enables an accurate characterization of spatially heterogeneous covert environments. We mathematically prove that a generalized on–off power control scheme, which jointly accounts for both Bob’s and Willie’s channels, constitutes the optimal transmission strategy in heterogeneous user configurations. Leveraging the optimal strategy, we derive closed-form expressions for the minimum detection error probability and the minimum number of cooperative users required to satisfy a covert constraint. With the closed-form expressions, comprehensive theoretical analyses are conducted, which are validated by Monte-Carlo simulations. One important insight obtained from the analysis is that user spatial heterogeneity can enhance covert communication performance. Building on these findings, a piecewise search algorithm is proposed to achieve exact optimality with significantly reduced computational complexity. We demonstrate that optimization considering user's spatial heterogeneity achieves substantially improved covert communication performance than that based on the assumption of spatial homogeneity.
\end{abstract}

\begin{IEEEkeywords}
Covert communication, physical layer security, piecewise optimization, spatial heterogeneity, wireless communications
\end{IEEEkeywords}

\ifCLASSOPTIONonecolumn
  \clearpage
  \pagenumbering{arabic}
\fi

\IEEEpeerreviewmaketitle
\section{Introduction}

With the explosive proliferation of wireless devices and the Internet of Things (IoT), modern wireless networks are required to manage enormous volumes of sensitive information. At the same time, the continuous growth in wireless connectivity and device density increases the exposure of communication channels, making the wireless medium highly vulnerable to adversarial surveillance and eavesdropping. To address these threats, \textit{physical layer security} (PLS) has emerged as a complementary paradigm to cryptography, exploiting the intrinsic randomness of wireless channel at the physical layer to provide confidentiality and resilience directly at the signal level. Among the various PLS techniques, \textit{covert communication} has recently attracted considerable attention for its unique objective. Unlike conventional cryptographic approaches that protect the content of a message, covert communication aims to conceal the existence of a transmission from adversarial detectors, typically called Willie. In other words, the transmitter seeks to ensure an extremely low probability that detectors can correctly decide whether a transmission is occurring, while simultaneously guaranteeing a desired data rate for the legitimate receiver. Such concealment capability is essential for diverse mission-critical and privacy-sensitive applications, including military operations, secure IoT networks, and covert ad hoc communications \cite{Jiang24Physical}.

Early covert transmission techniques, such as direct sequence spread spectrum, frequency hopping, and chirp spread spectrum, were primarily implementation-driven. Though they achieved practical covertness, the lack of theoretical analysis motivated a fundamental question on the limits of covert communication. In \cite{Bash13Limits}, Bash \textit{et al.} introduced the \textit{square root law} for additive white Gaussian noise (AWGN) channels. The law states that, under perfect channel knowledge at Willie, no more than $\mathcal{O}(\sqrt{n})$ bits can be transmitted reliably and covertly within $n$ channel uses. Subsequent works \cite{Che13Reliable,Bloch16Covert,Wang16Fundamental,Arumugam16Keyless,Arumugam19Covert,Tan19Time,Cho21Treating} generalized this principle to other canonical channels such as binary symmetric, discrete memoryless, multiple access, broadcast, and interference models, thus establishing the theoretical foundation of covert communication. 

Regrettably, strict adherence to the square root law dictates that the transmission rate must approach zero for a sufficiently large number of channel uses. Motivated by this limitation, numerous studies have explored ways to achieve a positive covert rate by introducing uncertainty at Willie. Major approaches to achieve positive covert rates include exploiting channel uncertainty~\cite{Lee15Achieving,Shahzad21Covert}, inducing interference via relays or jammers~\cite{Wang19Covert,Sun21Covert,Sun23Covertness,Bai22OnCovert,Sobers17Covert,Li20Optimal}, and utilizing full-duplex receivers~\cite{Shahzad18Achieving,Hu19Covert}. Furthermore, recent hardware-assisted innovations have significantly enhanced covertness by leveraging the inherent directionality of mmWave systems \cite{Zhang22Multi}, the spatial reconfigurability of movable antennas \cite{Mao25Sum}, and the mobility-enabled adaptability of UAV networks \cite{Jiang21Resource,Lei24Trajectory,Xu25Collaborative}.

Meanwhile, beyond covert communication in single-user wireless networks, covert communication has also been extensively studied in multi-user wireless networks. These studies analyzed the impact of random user locations, interference, and spatial topology on covert performance. He \textit{et al.} \cite{He18Covert} analyzed a homogeneous Poisson point process network in which numerous transmitters act as uncoordinated interferers. Their analyses demonstrated that, in the interference-limited regime, the covert rate converges to a finite bound regardless of the aggregate interference power, implying that background interference alone cannot indefinitely enhance covertness. Building upon this framework, Soltani \textit{et al.} \cite{Soltani18Covert} introduced friendly nodes that intentionally generate artificial noise. By activating the node closest to each Willie, the legitimate parties maintain Willie’s detection error probability near one half and achieves improved covertness through artificial noise. Although these studies provide valuable insights into the statistical behavior of large-scale stochastic networks, they rely on uncoordinated and statistically modeled interference. In particular, the interference sources are either naturally occurring, as in \cite{He18Covert}, or intentionally generated but still independent and non-cooperative, as in \cite{Soltani18Covert}. Consequently, these works do not investigate the potential of coordinated or jointly optimized user cooperation for actively enhancing covert performance.

More recent studies focus on \textit{cooperative} or \textit{active} multi-user covert communication, where legitimate users deliberately coordinate to confuse Willie \cite{Zheng21Wireless,Lee23Multi,Lee24Channel,Yeom24Covert}. Zheng \textit{et al.} \cite{Zheng21Wireless} proposed a distributed cooperative jamming scheme in which multiple friendly jammers are selected based on instantaneous channel gains to maximize the covert rate. Although the analysis primarily relies on numerical methods, their results demonstrate that cooperative diversity can significantly enhance covertness. Lee \textit{et al.}~\cite{Lee23Multi} further elaborated the advantages of cooperative diversity in covert communication. In particular, they derived closed-form expressions for the system performance that explicitly elucidate the relationships among system parameters, and proposed a low-complexity algorithm that maximizes the covert rate by jointly optimizing the transmit powers of the cooperative users and the covert message. Later, \cite{Lee24Channel} incorporated channel correlation, revealing that the presence of correlation among users can actually improve covertness, and proposed a scalable $Q$-learning based optimization algorithm. While these works have established the fundamental benefits of cooperation, they commonly assume homogeneous user configurations where all users are located at equal distances from both the legitimate receiver and Willie. To overcome the fundamental limitations of existing analyses for uplink multi-user covert communication, Yeom \textit{et al.} \cite{Yeom24Covert} investigated the covert communication performance in the uplink multi-user system where users are located at unequal distances from Bob. Although the users are spatially distributed, their mathematical analysis was derived based on single average values for the distances from the users to Bob and from the users to Willie, which substantially simplifies the analytical derivations. However, this oversimplifying assumption does not allow the analysis to capture the impact of users’ spatial heterogeneity on the covert communication performance.

Motivated by this limitation, this work focuses on how heterogeneous spatial configurations among multiple users affect the covert communication performance. Unlike the existing works \cite{Lee23Multi, Lee24Channel}, we consider an uplink multi-user covert communication network in which the user locations are spatially heterogeneous. In addition, we analyze the covert communication performance without relying on the oversimplifying assumption adopted in \cite{Yeom24Covert}, thereby enabling us to capture how spatially heterogeneous users affect the covert communication performance. The major contributions are summarized as follows:

\begin{enumerate}
    \item \textbf{Optimal power allocation strategy for covert communication:} 
    We mathematically prove that the on–off power allocation policy that accounts for all channels between the legitimate parties and Willie is optimal. This result reveals that the relative strength between the channel gain from each user to Bob and the variance of the channel from each user to Willie serves as a fundamental metric governing cooperative behavior. By analytically characterizing this ratio within a spatially heterogeneous network, the proposed framework provides the first geometric interpretation of how user locations influence the covert rate.
      
    \item \textbf{Closed-form expressions for the covert communication parameters:} 
    Closed-form analytical expressions are derived for key system parameters, including the minimum number of cooperative users and the activation threshold required to satisfy a convert constraint, and the covert rate. 
    These results yield explicit design guidelines and quantitative benchmarks for evaluating and optimizing covert performance under realistic spatially heterogeneous conditions.

    \item \textbf{Algorithmic efficiency:}
    By deriving all key parameters in the analytical forms, the proposed framework establishes explicit analytical relationships among system variables. Leveraging these relationships, we propose an algorithm offering high computational efficiency while maintaining analytical optimality within the heterogeneous multi-user framework.
\end{enumerate}
The subsequent sections of this work are detailed as follows. Section \ref{Sec:System} introduces the system model and Willie's detection strategy. In Section \ref{Sec:On-OffScheme}, the optimal transmit power profile for cooperative users is derived. Subsequently, Section \ref{Sec:DEP_Willie} derives closed-form expressions for the minimum number of cooperative users and the activation threshold required to satisfy the covert constraint. Building on these analyses, Section \ref{Sec:Optimization} proposes a piecewise search algorithm to efficiently search for the optimal covert message power. Finally, Section \ref{Sec:Simulation} validates the theoretical findings with numerical results, followed by the conclusion in Section \ref{Sec:conclusion}.

\section{System Model and Problem Formulation} \label{Sec:System}

\subsection{System Model} \label{Subsec:SystemModel}

\begin{figure}[t]
\centering
\includegraphics[width=\widefigwidth]{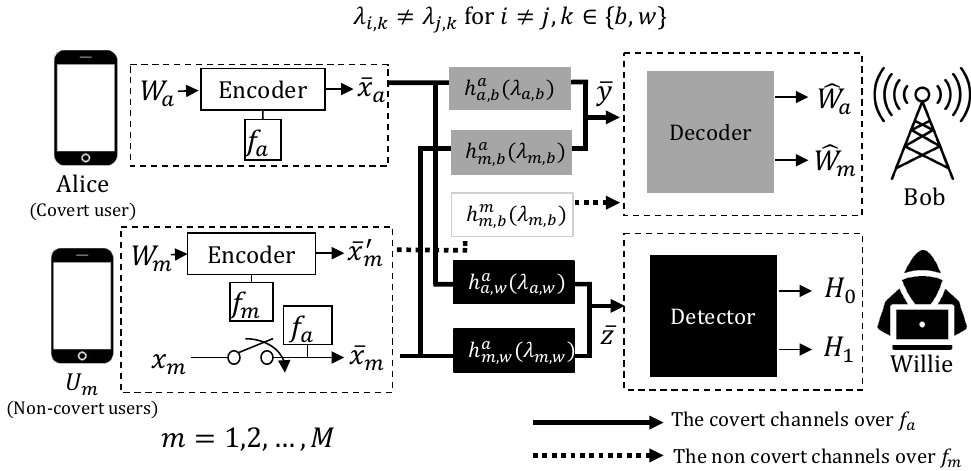}
\caption{System model of multi-user cooperation for covert communication.}
\label{fig:System Model} 
\end{figure}

We investigate uplink covert communication networks consisting of a covert user (Alice), $M$ non-covert users, a base station (Bob), and a watchful adversary (Willie), as illustrated in Fig. \ref{fig:System Model}. We assume that all users, including Alice and the non-covert users, communicate with Bob using frequency division multiple access. In particular, the $m$-th non-covert user, denoted by $U_m$, and Alice transmit data symbols $W_m$ and $W_a$ through carrier frequencies $f_m$ and $f_a$, respectively. A subset of non-covert users is selected as cooperative users to assist in hiding Alice’s signal from detection at Willie by transmitting interference over $f_a$. Hereafter, the term ``users'' refers to both the covert user and the non-covert users. We assume that all nodes including users and Bob, and Willie are equipped with a single antenna.

The channels are modeled as quasi-static Rayleigh fading, implying that the channel remains constant during a codeword transmission but varies independently across codewords. The channel from node $x$ to node $y$ at carrier frequency $z$ is denoted by 
\begin{equation}
 h_{x,y}^z \sim \mathcal{CN}(0,\lambda_{x,y}), \quad x,y \in \{ a,b,m,w \}, z \in \{ a,m \}   
\end{equation}
where $\lambda_{x,y}$ indicates the large-scale fading coefficient, and the indices $a,b,m$, and $w$ represent Alice, Bob, $U_m$, and Willie, respectively. In this work, the large-scale fading coefficient is assumed to be invariant regardless of the carrier frequency, and therefore the superscript indicating the carrier frequency is omitted from $\lambda_{x,y}$. We assume a single-slope path-loss model, and thus, the large-scale fading coefficient is expressed as
\begin{equation}
    \lambda_{x,y} = \beta_0 d_{x,y}^{-\alpha},
\end{equation}
where $\beta_0$ is the reference channel gain, $\alpha$ is the path-loss exponent, and $d_{x,y}$ denotes the Euclidean distances between nodes $x$ and $y$. Unlike the homogeneous channel assumption in \cite{Yeom24Covert, Lee23Multi, Lee24Channel}, where all users have a single large-scale fading coefficient, i.e., $\lambda_{m,w} = c$ for all $m$, this work considers on each user's distinct spatial position. Consequently, $\lambda_{m,w} \neq \lambda_{n,w}$ for $m \neq n$, capturing the spatial heterogeneity that sets our analysis apart from prior works.

Channel estimation process consists of two-phases. In the first phase, Bob broadcasts pilot signals over all frequency bands, i.e., $f_a$ and $\{ f_m \mid m \in \{ 1,2,\ldots, M\} \}$. Based on the received pilot signals, each node estimates the channel from Bob to itself. In the second phase, each user transmits its own pilot signal over its allocated frequency band. For covert communication, each non-covert user transmits a secret orthogonal pilot signal over $f_a$ \cite{Bash13Limits, Zheng21Wireless, Xu19Pilot}. The use of secret orthogonal pilot signals by each non-covert user enables Bob to estimate the channel between the non-covert user and Bob, whereas Willie cannot. Note that we assume perfect channel estimation; hence, the effect of channel estimation errors is neglected. In addition, channel reciprocity is assumed to hold, i.e., $h_{x,y}^z = h_{y,x}^z$ for all $x$ and $y$.

The $m$-th non-covert user, $U_m$, and Alice transmit codewords $\bar{x}'_m$ and $\bar{x}_a$ with length $N$ over $f_m$ and $f_a$, respectively, while each cooperative user also sends an interference sequence $\bar{x}_m$ over $f_a$ to help conceal $\bar{x}_a$. It is noteworthy that the interference signal, $\bar{x}_m$, also acts as interference at Bob, which degrades the covert communication performance. Therefore, the design of the interference power and the selection of cooperative users are of significant importance.  Our main interest lies in covert communication; hence, we only focus on the signal in $f_a$, where the covert signal is transmitted. The received signal at Bob $y[t]$ is given by
\begin{align} \label{eq:y_Bob}
    y[t]= 
  \begin{cases}
    \displaystyle\sum_{m = 1}^{M} \sqrt{P_m} h^a_{m, b} x_m[t] + n_b[t], & \mathcal{H}_0 \\
    \displaystyle\sqrt{P_a}h^a_{a, b} x_a[t] \\ ~~~~+ \sum_{m = 1}^{M} \sqrt{P_m}h^a_{m, b}x_m[t] + n_b[t], & \mathcal{H}_1
  \end{cases},
\end{align}
where both $x_a[t]$ and $x_m[t]$ are assumed to be complex Gaussain random variables with zero-mean unit-variance. The parameters $P_a$ and $P_m$ represent the transmit power of Alice and $U_m$, respectively, and $n_b[t] \sim \mathcal{CN}(0, \sigma_b^2)$ denotes AWGN at Bob. In addition, $\mathcal{H}_0$ and $\mathcal{H}_1$ indicate the hypotheses that Alice does not transmit or transmits a covert message, respectively. 
Similarly, the received signal at Willie over $f_a$ is given by
\begin{equation} \label{eq:y_Willie}
    z[t]= 
  \begin{cases}
      \displaystyle\sum_{m = 1}^{M} \sqrt{P_m} h^a_{m, w} x_m[t] + n_w[t], & \mathcal{H}_0 \\
      \sqrt{P_a}h^a_{a, w} x_a[t] \\ ~~~~+ \sum_{m = 1}^{M} \sqrt{P_m} h^a_{m, w} x_m[t] + n_w[t], & \mathcal{H}_1
  \end{cases},
\end{equation}
where $n_w[t] \sim \mathcal{CN}(0, \sigma_w^2)$ stands for AWGN at Willie.

\subsection{Detection at Willie}

Willie infers the presence of covert data transmission from Alice based on his observation. As introduced in Section \ref{Subsec:SystemModel}, Willie is confronted with a binary hypothesis testing problem consisting of two hypotheses $\mathcal{H}_0$ and $\mathcal{H}_1$. Under the conventional equal prior assumption \cite{Bash16Covert, Sobers17Covert}, the detection error probability (DEP) can be defined to
\begin{equation}
\zeta \triangleq P_{\mathrm{FA}} + P_{\mathrm{MD}},
\end{equation}
and the covert communication requirement is formalized through the constraint $\zeta \ge 1 - \epsilon$, where $\epsilon > 0$ is arbitrarily small~\cite{Hu19Covert, Lee23Multi, Lee24Channel}.

Since a detection method employed by Willie has a significant impact on the DEP, it is necessary to specify the detection method assumed in the analysis of covert communication performance. According to the Neyman–Pearson Lemma \cite{Neyman33OntheProblem}, the likelihood ratio test (LRT) provides the optimal detection performance for the binary hypothesis test and is expressed as 
\begin{equation}
    \frac{\textstyle\prod_{t=1}^{N} f\bigl(z[t]; \mathcal{H}_1\bigr)}
           {\textstyle\prod_{t=1}^{N} f\bigl(z[t]; \mathcal{H}_0\bigr)}
    \mathop{\lessgtr}_{\mathcal{D}_1}^{\mathcal{D}_0} \gamma,
\end{equation}
where $f(\cdot)$ represents the probability density function (PDF) and $\gamma$ denotes the detection threshold. It is noteworthy that when the transmitted signals follow a Gaussian distribution, the received energy, i.e., $T_w = \sum_{t=1}^{N} |z[t]|^2 /N$, serves as a sufficient statistic for the LRT \cite{Li20Optimal, Zheng21Wireless}. Consequently, the LRT can be equivalently reduced to a test of the received energy, i.e., $T_w$ as follows:

\begin{equation} \label{eq:test}
T_w \mathop{\lessgtr}_{\mathcal{D}_1}^{\mathcal{D}_0} \gamma.
\end{equation}
Since $x_a[t], x_m[t]$ and $n_w[t]$ are uncorrelated Gaussian random variables, the distribution of $T_w$ under both hypotheses follows a chi-square distribution, as given by
\begin{align}
  T_w =
\begin{cases}
  \displaystyle\left(\sum_{m=1}^{M} P_m |h^a_{m,w}|^2 + \sigma_w^2 \right) \frac{\chi^2_{2N}}{2N}, & \mathcal{H}_0 \\[8pt]
  \displaystyle\left(\sum_{m=1}^{M} P_m |h^a_{m,w}|^2 + P_a |h^a_{a,w}|^2 + \sigma_w^2 \right) \frac{\chi^2_{2N}}{2N}, & \mathcal{H}_1
\end{cases}  
\end{align}
where $\chi^2_{2N}$ denotes a chi-square random variable with $2N$ degrees of freedom. When the length of codeword, $N$, is sufficiently large, the ratio $\chi^2_{2N}/(2N)$ converges to unity due to the strong law of large numbers \cite{Casella02StatInference}. Accordingly, the test statistic $T_w$ can be characterized as
\begin{align}\label{eq:test_statistic}
   T_w =
  \begin{cases}
    \displaystyle\sum_{m=1}^{M} P_m |h^a_{m,w}|^2 + \sigma_w^2, & \mathcal{H}_0 \\[4pt]
    \displaystyle\sum_{m=1}^{M} P_m |h^a_{m,w}|^2 + P_a |h^a_{a,w}|^2 + \sigma_w^2, & \mathcal{H}_1
  \end{cases}.
\end{align}
Then, the DEP, $\zeta(\gamma)$, for a given threshold $\gamma$ is given by
\begin{equation} \label{eq:DEP}
  \zeta(\gamma) = \Pr(T_w > \gamma \mid \mathcal{H}_0) + \Pr(T_w \leq \gamma \mid \mathcal{H}_1).
\end{equation}
Note that $\zeta(\gamma)$ depends on the detection threshold $\gamma$. To consider the worst case, we assume that Willie chooses the threshold $\gamma$ that minimizes the DEP, i.e.,
\begin{equation}
    \zeta_{\min} = \min_{\gamma} \zeta(\gamma).
\end{equation}
It is shown in \eqref{eq:test_statistic} that Willie faces uncertainty caused by (i) his access only to the statistical properties of the channel, $|h_{m,w}|^2$, and (ii) the randomness of the aggregated interference power from spatially distributed users, $\small\sum_{m=1}^{M} P_m |h^a_{m,w}|^2$. Consequently, the spatial heterogeneity of user locations becomes a key factor determining covert communication performance.

\subsection{Optimization Problem}

This work focuses on the transmit power optimization problem to maximize the achievable rate while satisfying the covert constraint. The achievable rate is defined as
\begin{equation}\label{eq:rate_def}
R = \log_2\left(1 + \frac{P_a |h^a_{a,b}|^2}{\sum_{m=1}^M P_m |h^a_{m,b}|^2 + \sigma_b^2}\right).
\end{equation}
Then, the transmit power optimization problem can be formulated as
\begin{subequations}\label{eq:original_opt}
\begin{align}\label{eq:rate_opt}
&(P_a^*, \bar{P}^*) = \argmax_{P_a, \bar{P}} ~ R \\
& ~~~~\text{s.t.} ~~ \zeta_{\min} \geq 1- \epsilon \label{eq:CovertConst}\\
& ~~~~~~~~~ 0 \leq P_a \leq P_{\max} \label{eq:PowConst_Alice}\\
& ~~~~~~~~~ 0 \leq P_m \leq P_{\max}, \quad \forall m, \label{eq:PowConst_Noncovert}
\end{align}
\end{subequations}
where $\bar P \triangleq (P_1,P_2,\ldots,P_M)$ is the transmit power profile for non-covert users, $P_{\max}$ is the maximum transmit power, and $\epsilon$ is the tolerable covertness error. The covert constraint is given in \eqref{eq:CovertConst}, and the transmit power constraints are given in \eqref{eq:PowConst_Alice} and \eqref{eq:PowConst_Noncovert}, respectively.

\section{Optimal Transmit Power Profile} \label{Sec:On-OffScheme}

In this section, the optimization problem in \eqref{eq:original_opt} is reformulated to a more tractable and simpler form by analyzing the optimal transmit power profile $\bar{P}$ for a given $P_a$. 

\begin{prop}\label{prop:power_profile}
For sufficiently large $M$, the optimal transmit power profile for non-covert users under spatial heterogeneity is an on–off scheme for a given $P_a$, i.e., 
\begin{equation}
       P_m^*= 
  \begin{cases}
      P_{\max}, & \dfrac{|h^a_{m,b}|^2}{\lambda_{m,w}} \leq \tau \\
      0 , & \text{otherwise}.
  \end{cases},
\end{equation}
\end{prop}
\begin{proof}
The proof is provided in Appendix \ref{Appendix:power_profile}.
\end{proof}

Based on the on-off scheme presented in Proposition \ref{prop:power_profile}, we reformulate the original optimization problem in \eqref{eq:original_opt} as 
\begin{subequations} \label{eq:reformul_opt}
\begin{equation} 
  \left(P_a^*, \tau^* \right)=\argmax_{P_a, \tau} ~ R 
   ~~~~ \text{s.t.} ~~ \eqref{eq:CovertConst} \text{ and } \eqref{eq:PowConst_Alice}. 
  \end{equation}
\end{subequations}
Note that the reformulated optimization problem in \eqref{eq:reformul_opt} has only two optimization variables, whereas the original optimization problem in \eqref{eq:original_opt} involves $M+1$ optimization variables. Consequently, owing to Proposition \ref{prop:power_profile}, a much simpler optimization problem can be obtained. 

\begin{remark}
    In this work, Bob optimizes the communication parameters, $P_a$ and $\bar{P}$ (or equivalently, the activation threshold $\tau$), to maximize the achievable rate for the covert message while satisfying both the covert and power constraints, given the instantaneous channel gains $\{|h_{x,b}^a| : x \in \{ a,m\} \}$ and the large-scale fading coefficients $ \{ \lambda_{x, b} : x \in \{ a,m\} \}$. The optimized parameters, i.e, $P_a$ and $\tau$, are then broadcast to users and Willie. 
\end{remark}
\begin{remark}\label{remark:Determination_K}
Intuitively, increasing the activation threshold $\tau$ allows more non-covert users to transmit interference signals. This enhances the concealment of the covert message but simultaneously degrades Bob’s communication performance due to the increased interference power. In other words, there exists a fundamental trade-off between covertness and reliable communication with respect to $\tau$. Consequently, the optimal threshold $\tau$ is selected as the minimum value that satisfies the covert constraint $\zeta_{\min} \geq 1 - \epsilon$ in~\eqref{eq:CovertConst}; that is, the optimal $\tau$, denoted by $\tau^*$, makes the covert constraint hold with equality, i.e., $\zeta_{\min}(\tau^*) = 1 - \epsilon$. Equivalently, since $\tau$ determines the number of cooperative users $K$ through the condition $|h^a_{m,b}|^2 / \lambda_{m,w} \leq \tau$, the optimal $K$ is also chosen as the minimum number of cooperative users that satisfies the covert constraint.
\end{remark}

\begin{remark}
In the existing studies \cite{Lee23Multi,Lee24Channel,Yeom24Covert}, it has been mathematically proven that the optimal power profile follows the on–off scheme that considers only the instantaneous channel gains between the non-covert users and Bob, i.e., $|h^a_{m,b}|$. In contrast, the on–off scheme proposed in this work takes into account not only $|h^a_{m,b}|$, but also the large-scale fading coefficients between the non-covert users and Willie, i.e., $\lambda_{m,w}$. This reveals a key distinction from the previous on-off scheme \cite{Lee23Multi, Lee24Channel, Yeom24Covert}, which aim only to reduce interference at Bob without accounting for Willie. In contrast, the on–off scheme proposed in this work explicitly incorporates Bob’s instantaneous channel gains, i.e., $|h_{m,b}^a|$, and Willie’s large-scale fading coefficients, i.e., $\lambda_{m,w}$, thereby selecting cooperative users that reduce interference at Bob while increasing interference observed by Willie.
\end{remark}

\begin{remark}\label{remark:Info_Willie}
In this work, we adopt a worst-case adversary model for Willie. Specifically, Willie is assumed to have perfect knowledge of the on–off scheme and the total number of cooperative users $K$. Unlike~\cite{Yeom24Covert}, which assumes a single averaged parameter $\bar{\lambda} = \mathbb{E}[\lambda_{m,w}]$ due to Willie’s incomplete knowledge of user locations, we assume that Willie knows the exact locations of all users. This enables him to compute the individual large-scale fading coefficients $\lambda_{m,w}$ for every $m$, ensuring that his detection performance fully exploits the available spatial information. Such a worst-case modeling choice leads to a conservative and robust secrecy analysis.
\end{remark}

\section{Analysis of Detection at Willie} \label{Sec:DEP_Willie}

In this section, we derive a closed-form expression of the minimum DEP and analyze its characteristics. In particular, we determine the optimal detection threshold $\gamma^*$ that minimizes the DEP, resulting in the closed-form expression of $\zeta_{\min}$. Furthermore, based on the closed-form expression of $\zeta_{\min}$, we derive the minimum number of cooperative users and the corresponding activation threshold required to satisfy the covert constraint, i.e.,
$\zeta_{\min} \ge 1-\epsilon$, both in closed-form expressions.

According to the on-off scheme in Proposition \ref{prop:power_profile}, the cooperative users transmit interference signals with $P_m = P_{\rm{max}}$, while the other users do not transmit any interference signals, i.e., $P_m = 0$. Then, under the on-off scheme in Proposition \ref{prop:power_profile}, the test statistic in \eqref{eq:test_statistic} is expressed as 
\begin{align}\label{eq:test_statistic2}
   T_w = 
  \begin{cases}
    \displaystyle\sum_{i=1}^{K} P_{\max} |h^a_{m_i,w}|^2 + \sigma_w^2, & \mathcal{H}_0, \\[4pt]
    \displaystyle\sum_{i=1}^{K} P_{\max} |h^a_{m_i,w}|^2 + P_a |h^a_{a,w}|^2 + \sigma_w^2, & \mathcal{H}_1,
  \end{cases}
\end{align}
where $m_i \in \{1,2,\ldots, K\}$ denotes the index of the $i$-th cooperative user sorted in ascending order of the ratio between the channel gain and the large-scale fading coefficient, i.e., $|h^{a}_{m_i,b}|^2 / \lambda_{m_i,w} \ge |h^{a}_{m_j,b}|^2 / \lambda_{m_j,w}$ for $i \ge j$.

Since each channel gain $|h^{a}_{m_i,w}|^2$ follows a Gamma distribution, and the sum of independent random variables is represented by the convolution of their distributions, the shifted test statistic $T_w - \sigma_w^2$ 
is expressed as
{\small
\begin{align}\label{eq:test_statistic_RV}
  T_w - \sigma_w^2 \sim
   \begin{cases}
      \displaystyle\Conv_{i=1}^K \Gamma(1, \lambda_{m_i,w} P_{\max}), & \mathcal{H}_0, \\[4pt]
      \displaystyle\Conv_{i=1}^K \Gamma(1, \lambda_{m_i,w} P_{\max}) \ast \Gamma(1, \lambda_{a,w}P_a), & \mathcal{H}_1,
   \end{cases}
\end{align}
} 
\par \noindent where $\Gamma(a,b)$ represents a Gamma distribution with shape parameter $a$ and scale parameter $b$, $\ast$ denotes the convolution operator between PDFs of two random variables, and $\Conv_{i=1}^K$ indicates the successive convolution of the PDFs of $K$ random varaibles, i.e., $\Conv_{k=1}^K A_k = A_1 \ast A_2 \ast \cdots \ast A_K$. Based on the probabilistic distribution of $T_w - \sigma_w^2$ in \eqref{eq:test_statistic_RV}, a closed-form expression for the DEP in \eqref{eq:DEP} is derived in Lemma~\ref{lemma:DEP}.

\begin{lemma}\label{lemma:DEP}
With the on-off scheme, the closed-form expression for the DEP $\zeta (\gamma)$ is approximated as  
\begin{align}\label{eq:DEP_general}
  \zeta (\gamma) 
  &\approx 1 - 
  \exp\!\left(
    -\frac{2\Delta(\hat{\gamma}-\Xi) - \Sigma}{2\Delta^2}
  \right)
  Q\!\left(
    -\frac{\Delta(\hat{\gamma}-\Xi) - \Sigma}{\sqrt{\Sigma}\,\Delta}
  \right),
\end{align}
where $\hat{\gamma} \triangleq \gamma - \sigma_w^2$, $\Delta \triangleq P_a \lambda_{a,w}$, $\Xi \triangleq \frac{KP_{\max}}{M} \sum_{m=1}^{M}\lambda_{m,w}$, and
\begin{align}
  \Sigma &\triangleq P_{\max}^2 \Bigg[
   \frac{K}{M}\sum_{m=1}^{M}\lambda_{m,w}^{2}
   + \frac{K(M-K)}{M-1}\\ 
   & \hspace{5em} \times \frac{1}{M} \sum_{m=1}^{M}
     \left(
       \lambda_{m,w}
       - \frac{1}{M}\sum_{n=1}^{M}\lambda_{n,w}
     \right)^{2}
   \Bigg].  
\end{align}
\end{lemma}

\begin{proof}
See Appendix \ref{Appendix:DEP}
\end{proof}

The analytical form of the DEP in \eqref{eq:DEP_general} enables one to obtain several design parameters in closed-form. In particular, the optimal detection threshold $\gamma^\ast$, which minimizes the DEP, is derived in closed-form and presented in Lemma \ref{lemma:gamma}.
\begin{lemma}\label{lemma:gamma}
The optimal detection threshold is given by 
\begin{equation}\label{eq:Opt_gamma}
  \gamma^\ast = \frac{K P_{\max}}{M} \sum_{m=1}^M \lambda_{m,w} + \sigma_w^2.
\end{equation}
\end{lemma}

\begin{proof}
Starting from \eqref{eq:DEP_general}, the DEP is approximated as
\begin{equation}
  \zeta (\gamma) \approx
  1 - c_q \exp\left( \underbrace{  -\frac{1}{2\Sigma} \left( \hat\gamma^2-2\Xi\hat\gamma + \Xi^2 \right)}_{(a)}  \right), \label{eq:ApproxDEP}
\end{equation}
where the approximation follows from the fact that the $Q$-function can be approximated by an exponential function, i.e., $Q(x)\approx c_q e^{-x^2/2}$ \cite{Lee23Multi}. It is observed that the DEP expression in \eqref{eq:ApproxDEP} is an exponential function of a quadratic term in $\hat{\gamma}$. Since the exponential function is monotonically increasing function and the quadratic function with a negative quadratic coefficient, i.e., the term (a) in \eqref{eq:ApproxDEP}, attains its maximum value at the vertex $\hat{\gamma} = \Xi$, the minimum DEP is achieved at $\hat{\gamma}^\ast = \Xi$. Based on the definitions of $\hat{\gamma}$ and $\Xi$, the optimal detection threshold is given in \eqref{eq:Opt_gamma}.
\end{proof}

Using the optimal detection threshold in \eqref{eq:Opt_gamma}, the minimum DEP is obtained as presented in Lemma \ref{lemma:Min_DEP}.
\begin{lemma}\label{lemma:Min_DEP}
The minimum DEP at Willie is given by
\begin{equation}\label{eq:Opt_DEP}
  \zeta_{\min}
  = 1 - \left( \sqrt{\pi}
  \left(
  \sqrt{\tfrac{\Sigma}{2\Delta^2}}
  + 
  \sqrt{\tfrac{\Sigma}{2\Delta^2} + \tfrac{4}{\pi}}
  \right) \right)^{-1}.
\end{equation}
\end{lemma}

\begin{proof}
By substituting the optimal detection threshold $\gamma^\ast$ in \eqref{eq:Opt_gamma} into the DEP in \eqref{eq:DEP_general}, we obtain
\begin{align} 
    \zeta_{\min} &= 1-\exp\!\left(\frac{\Sigma}{2\Delta^2}\right)Q\!\left(\frac{\sqrt{\Sigma}}{\Delta}\right) \nonumber\\
    &  = 1 - \frac{1}{2}\exp\!\left(\frac{\Sigma}{2\Delta^2}\right)\mathrm{erfc}\!\left(\frac{\sqrt{\Sigma}}{\sqrt{2}\,\Delta}\right) \nonumber \\
    & \mathop{\approx}_{{\mathrm{(a)}}} 1 - \frac{1}{2}\exp \left(\frac{\Sigma}{2\Delta^2} \right)\frac{2\exp \left( -\frac{\Sigma}{2\Delta^2} \right) }{\sqrt{\pi}\left(\sqrt{\frac{\Sigma}{2\Delta^2}} + \sqrt{\frac{\Sigma}{2\Delta^2} + \frac{4}{\pi}} \right)} \nonumber \\
    & = 1-\left( \sqrt{\pi}\left(\sqrt{\frac{\Sigma}{2\Delta^2}} + \sqrt{\frac{\Sigma}{2\Delta^2} + \frac{4}{\pi}} \right) \right) ^{-1} ,
\end{align}
where $\mathrm{erfc}(\cdot)$ indicates the complementary error function. The approximation in (a) follows from the fact that 
$\operatorname{erfc}(x) \approx 2/\sqrt{\pi} \cdot \operatorname{exp}(-x^2)/(x + \sqrt{x^2 + 4/\pi})$ is highly accurate for $x \gtrsim 1$ \cite{Whittaker90ACourse}. In our setting, the argument of the complementary error function is $\Sigma / 2\Delta^2$, and since the aggregate interference power $\Sigma$ from the cooperative users is typically much larger than that of Alice's power $\Delta$, this approximation is well justified. 
\end{proof}

As discussed in Remark \ref{remark:Determination_K}, the number of cooperative users $K$ is determined as the minimum integer value satisfying the covert constraint $\zeta_{\min} \ge 1-\epsilon$. Based on this discussion, the minimum number of cooperative users, denoted by $K_{\min}$, is derived in Theorem~\ref{Theorem:K}.

\begin{theorem}\label{Theorem:K}
The minimum number of cooperative users is given by
\begin{equation}\label{eq:K1_exact_approx}
  \displaystyle K_{\min}= \left\lceil \frac{M(E+V)}{2V}\left( 1 -\sqrt{ 1 -\frac{4V}{CM(E+V)^2}}\right) \right\rceil,
\end{equation}
where $C=P_{\max}^2 / \left(P_a^2 \lambda^2_{a,w} c_\epsilon \right)$ with $c_\epsilon = (1/\epsilon^2 - 8 + 16\epsilon^2) / 2\pi$. In addition, $E$ means the empirical mean of $\lambda_{m,w}^2$, which is defined by
\begin{equation}
 E \triangleq \mathbb{E}[\lambda_w^2] = \frac{1}{M}\sum_{m=1}^{M}\lambda_{m,w}^2,    
\end{equation}
and $V$ represents the empirical variance of $\lambda_{m,w}$, which is expressed by 
\begin{equation}
    V \triangleq \mathrm{Var}[\lambda_w] = \frac{1}{M}\sum_{m=1}^{M}\!\left(\lambda_{m,w}-\frac{1}{M}\sum_{n=1}^{M}\lambda_{n,w}\right)^{2}.
\end{equation}

\end{theorem}

\begin{proof}
See Appendix \ref{Append:ProofMinK}
\end{proof}

Although the expression for $K_{\min}$ in \eqref{eq:K1_exact_approx} has a simple analytical form, it is somewhat challenging to directly observe how the communication parameters affect it. For this reason, we further analyze two special cases: 
1) the asymptotic case, i.e., $M \gg 1$, in Corollary~\ref{Corollary:K_approx}, and 
2) the homogeneous user deployment case, i.e., $\lambda_{m,w} = \lambda_{n,w}$ for $m \neq n$, in Corollary~\ref{cor:homogeneous}.

\begin{corollary}\label{Corollary:K_approx}
For $M\gg1$, $K_{\min}$ is approximated as
\begin{equation}\label{eq:K_simple_approx}
  \displaystyle K_{\min} \approx 
  \left\lceil\frac{1}{C(E+V)}\right\rceil
  =\left\lceil
  \frac{(P_a\lambda_{a,w})^2 c_\epsilon}
  {P_{\max}^2\bigl(\mathbb{E}[\lambda_w^2]+\mathrm{Var}[\lambda_w]\bigr)}
  \right\rceil.
\end{equation}
\end{corollary}

\begin{proof}
Let us focus on the square-root term in \eqref{eq:K1_exact_approx}. By applying the first-order binomial expansion, i.e., $\sqrt{1-x} \approx 1 - x/2$, the square-root term can be approximated as
    \begin{equation}
        \sqrt{ 1 - \frac{4V}{CM(E+V)^2}} \approx 1-\frac{2V}{CM(E+V)^2}. \label{eq:MinKAsymp}
    \end{equation}
Note that the first-order binomial expansion is accurate when $x$ is small. 
In this work, $x$ corresponds to the term $4V / (CM(E + V)^2)$. When $M$ increases, the values of $C$, $E$, and $V$ do not scale with $M$, so the term $4V / (CM(E + V)^2)$ indeed becomes small for sufficiently large $M$. This implies that when $M$ is sufficiently large, the approximation in \eqref{eq:MinKAsymp} is justified. Then, by substituting the approximation in \eqref{eq:MinKAsymp} into \eqref{eq:K1_exact_approx}, and performing some straightforward mathematical manipulations, we obtain the expression of $K_{\min}$ presented in \eqref{eq:K_simple_approx}.
\end{proof}

From the approximation of $K_{\rm{min}}$ presented in \eqref{eq:K_simple_approx}, we can obtain some insight into how the primary factors influence the required $K_{\min}$. In particular, the result indicates that $K_{\min}$ is inversely proportional to $\mathbb{E}[\lambda_w^2] + \mathrm{Var}[\lambda_w]$, and increases quadratically with $P_a$. In addition, $K_{\min}$ increases on the order of $\epsilon^{-2}$ as $\epsilon$ becomes small, because $c_{\epsilon}$ itself scales as $\mathcal{O}(\epsilon^{-2})$ as $\epsilon \to 0$. This implies that as the covert constraint becomes tighter, that is, as $\epsilon$ decreases, the minimum number of cooperative users increases on the order of the inverse square of $\epsilon$. All these analyses are confirmed in Section \ref{Sec:Impactofepsilon}.

\begin{corollary}\label{cor:homogeneous}
Under homogeneous user deployment, where all non-covert users are located at approximately the same distance from Willie, i.e., $\lambda_{m,w} \approx \bar{\lambda}_w$ for all $m$ and thus $\mathrm{Var}[\lambda_w]\!\approx\!0$, but the Alice--Willie channel gain $\lambda_{a,w}$ may differ from $\bar{\lambda}_w$, the approximation expression of $K_{\rm min}$ in Theorem \ref{Theorem:K} simplifies to
\begin{equation}
K_{\min} \approx \left\lceil
\frac{P_a^2 c_\epsilon}{P_{\max}^2}
\left( \frac{\lambda_{a,w}}{\bar{\lambda}_w} \right)^{\!2}
\right\rceil.
\label{eq:Kmin_simplified}
\end{equation}
\end{corollary}

\begin{proof}
For simplicity, the ceiling function in \eqref{eq:K1_exact_approx} is omitted. Under homogeneous user deployment, the term $V$ is very small; thus, using the first-order binomial approximation $\sqrt{1-x}\approx 1-x/2$ for $x = 4V/(CM(E+V)^2)$ in Theorem~\ref{Theorem:K} yields a highly accurate approximation. After some straightforward algebraic manipulation, $K_{\min}$ can be written as
\begin{equation}
K_{\min} \approx \frac{1}{C(E+V)} \mathop{\approx}_{{\mathrm{(a)}}} \frac{1}{CE},
\end{equation}
where (a) follows from $V \ll E$ in the homogeneous case.  Finally, substituting $C = P_{\max}^2 / \bigl( (P_a\lambda_{a,w})^2 c_\epsilon \bigr)$ and applying 
$E=\bar{\lambda}_{w}^{\,2}$, completes the proof.
\end{proof}
In this case, the minimum number of cooperative users is independent of the number of non-covert user $M$. In addition, when $\lambda_{a,w} = \bar{\lambda}_w$, this result directly reduces to the result shown in previous works \cite{Lee23Multi, Yeom24Covert}, demonstrating that the proposed analysis encompasses these prior studies as special cases and therefore provides a more general characterization of the system.

By utilizing the closed-form expression for $K_{\min}$ in Theorem \ref{Theorem:K}, the corresponding optimal activation threshold of the on-off scheme in Proposition \ref{prop:power_profile} can be readily derived, as presented in Theorem \ref{theorem:tau}.

\begin{theorem}\label{theorem:tau}
The optimal activation threshold $\tau^*$ is determined by
\begin{equation}
    \tau^* = \frac{|h^{a}_{K_{\min},b}|^{2}}{\lambda_{K_{\min},w}}.
\end{equation} 
\end{theorem}

\begin{proof}
By setting $\tau$ to the ratio $|h^{a}_{K_{\min},b}|^{2}/\lambda_{K_{\min},w}$, the $K_{\min}$-th weakest user in terms of $|h^{a}_{m,b}|^{2} / \lambda_{m,w}$ becomes the activation boundary, ensuring that exactly $K_{\min}$ cooperative users participate, where $K_{\min}$ is given in \eqref{eq:K1_exact_approx}.
\end{proof}

\section{Optimization of System Parameters}\label{Sec:Optimization}

In this section, we first reformulate the optimization problem in \eqref{eq:reformul_opt}, which involves two optimization variables, into a simpler form with only a single optimization variable based on Theorems~\ref{Theorem:K} and~\ref{theorem:tau}. Then, we propose an algorithm that efficiently solves the reformulated optimization problem.

\subsection{Reformulation of the Optimization Problem}

As discussed in Section \ref{Sec:On-OffScheme}, adopting the on–off scheme in Proposition \ref{prop:power_profile} allows the original problem in \eqref{eq:original_opt} to be reformulated as \eqref{eq:reformul_opt}, reducing the number of optimization variables from $M+1$, i.e., $P_a, P_1, \ldots, P_M$, to two, $P_a$ and $\tau$. For a given $P_a$, Theorems \ref{Theorem:K} and \ref{theorem:tau} characterize the optimal number of cooperative users and the corresponding activation threshold. Consequently, the problem in \eqref{eq:reformul_opt} can be equivalently rewritten as

\begin{equation} \label{eq:master_opt}
    P_a^* = \argmax_{0 \le P_a \le P_{\max}} R(P_a,K_{\rm{min}}),
\end{equation}

where 
\begin{equation}
   \textstyle R(P_a,K_{\rm{min}}) = \log_2\left(1 + \frac{P_a |h^a_{a,b}|^2}{\sum_{m=1}^{K_{\rm{min}}} P_{\rm{max}} |h^a_{m_i,b}|^2 + \sigma_b^2}\right). \label{eq:Rate_OnOff}
\end{equation}
This optimization problem has only one optimization variable, i.e., $P_a$, and is therefore much easier to solve than the original formulation of \eqref{eq:original_opt}, which involves $M+1$ optimization variables.

\subsection{Optimization of Transmit Power for Alice}
Notably, as shown in Theorem \ref{Theorem:K} and Corollaries \ref{Corollary:K_approx} and \ref{cor:homogeneous}, $K_{\min}$ is a function of $P_a$. In particular, Corollaries \ref{Corollary:K_approx} and \ref{cor:homogeneous} explicitly demonstrates that $K_{\min}$ increases with $P_a$. Intuitively, a larger $P_a$ requires stronger interference, that is, a greater $K_{\min}$, to satisfy the covert constraint. Therefore, the optimal $P_a$ in \eqref{eq:master_opt} is not simply the maximum allowable value $P_{\max}$. In this work, we propose an algorithm that determines the optimal $P_a$ through a piecewise search for the optimization problem in \eqref{eq:master_opt}.

Meanwhile, it is noteworthy that $K_{\min}$ in \eqref{eq:K1_exact_approx} includes a ceiling function. As a result, $K_{\min}$ remains unchanged over certain ranges of $P_a$. In such regime where $K_{\min}$ is constant, the optimal $P_a$ is the maximum value within that range. The maximum value of $P_a$ corresponding to a given $K_{\min}=\bar K$ can be obtained from Theorem \ref{Theorem:K} as

\begin{equation}
\label{eq:Pa_of_K_exact_opt}
P_a(\bar K)=\frac{P_{\max}}{\sqrt{c_\epsilon}\,\lambda_{a,w}}
\sqrt{(E+V)\bar K-\frac{V}{M}(\bar K)^2}.
\end{equation}
Then, for a given $K_{\min}=\bar K$, the optimal $P_a$ is defined as $P_{\bar K} = \min\big(P_{\max},\, P_a(\bar K)\big)$ for $\bar K \in \{0,\ldots,M\}$. Thus, when performing a one-dimensional search for the optimal $P_a$, it is sufficient to partition the feasible interval of $P_a$ according to the values of $P_{\bar K}$, which naturally leads to a piecewise search structure. This piecewise characterization eliminates the need to explore redundant regions of the search space, making the resulting piecewise search far more efficient and structure-aware than a conventional exhaustive one-dimensional search. The overall piecewise search is summarized in Algorithm \ref{alg:piecewise}.

\begin{algorithm}[t] 
\caption{Piecewise Search} 
\ifCLASSOPTIONonecolumn\small\fi
\begin{algorithmic}[1] \label{alg:piecewise}
\renewcommand{\algorithmicrequire}{\textbf{Input:}}
\renewcommand{\algorithmicensure}{\textbf{Output:}}
\REQUIRE  $h^a_{a,b}$, $\{h^a_{m,b}\}_{m=1}^{M}$, $\{\lambda_{m,w}\}_{m=1}^{M}$, $P_{\max}$, $\epsilon$, $\sigma_b^{2}$
\ENSURE $P_a^{*}$, $K^{*}$, $\tau^{*}$
\STATE Compute $C$, $E$, $V$, and $c_\epsilon$, respectively.
\STATE Set $R_{\text{temp}} \leftarrow 0$.
\FOR{$\bar K=0$ to $M$}
  \STATE Compute $P_{\bar K}$ and set $P_a\leftarrow P_{\bar K}$.
  \STATE Based on $K_{\min} = \bar K$, compute $R(P_a,\bar K)$ in \eqref{eq:Rate_OnOff}.
  \IF{$R(P_a, \bar K) > R_{\text{temp}}$}
    \STATE $R_{\text{temp}}\leftarrow R(P_a, \bar K)$, $P_a^{*}\leftarrow P_a$, and $K^{*}\leftarrow K$.
  \ENDIF
\ENDFOR
\STATE Set $\tau^{*}\leftarrow \bigl|h^a_{m_{K^{*}},b}\bigr|^{2}/\lambda_{m_{K^{*}},w}$ and return $P_a^{*}$, $K^{*}$, $\tau^{*}$.
\end{algorithmic}
\end{algorithm}

\section{Numerical Results} \label{Sec:Simulation}

In this section, comprehensive numerical evaluations are carried out to validate the theoretical analyses derived in this work. In addition, we examine the performance of the proposed covert communication network where users are spatially distributed and reveal the interactions among key system parameters.

\subsection{Simulation Setup} \label{Subsec:simulation_setup}

For performance evaluations, we consider the large-scale fading channel model specified in the 3GPP TR 25.996 Urban Macrocell guideline~\cite{Ref_3GPP_TR25996}. The path loss (PL) at a distance $d$ (in meters) is modeled as
\begin{equation}
    \text{PL(dB)} = 34.5 + 35\log_{10}(d),
\end{equation}
which corresponds to the large-scale fading coefficient $\lambda_{x,y} = \beta_0 d_{x,y}^{-\alpha}$. Here, the path-loss exponent is set to 3.5, and the reference channel gain at $d=1\mathrm{m}$ is given by $\beta_0=-34.5 \mathrm{dB}$. The maximum transmit power of each user is set to $P_{\max} = 200\,\mathrm{mW}\,(23\,\mathrm{dBm})$, and the noise variances at Bob and Willie are assumed to be identical, i.e., $\sigma_b^2 = \sigma_w^2 = -102\,\mathrm{dBm}$.

The spatial topology of the network is constructed over a circular region with a radius of approximately $707\,\mathrm{m}$, representing the \textbf{spatially adverse scenario} illustrated in Fig.~\ref{fig:Simulation_Setup}. Willie is placed at the center of the area, $(500,500)$, while Bob is located near the boundary at $(100,100)$. Alice is positioned at $(831,831)$ such that her distance to Willie is $d_{a,w}\approx468\,\mathrm{m}$. This layout intentionally enforces two defining characteristics of this adverse configuration. First, the distance from Alice to Bob is significantly larger than that to Willie ($d_{a,b} \gg d_{a,w}$), forcing Alice to transmit with relatively higher power to sustain a reliable link, which simultaneously strengthens the received signal at the nearby adversary, i.e., Willie. Second, all $M=1000$ non-covert users are uniformly distributed within an annular region satisfying $d_{m,w}>d_{a,w}$ for every $m\in\{1,\dots,M\}$. Consequently, each non-covert users experiences a weaker channel to Willie compared to Alice ($\lambda_{m,w}<\lambda_{a,w}$), minimizing the received interference power and imposing the most unfavorable conditions for covertness. Under this setting, the system must satisfy the covert constraint using only these disadvantaged non-covert users, thereby providing a rigorous and conservative benchmark for the proposed scheme.

\begin{figure}[t]
\centering
\includegraphics[width=\figwidth]{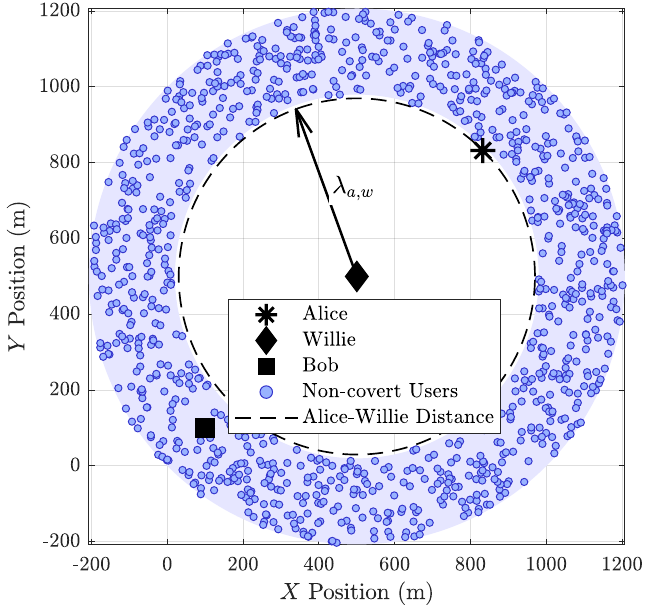}
\caption{Spatial configuration of the spatially adverse scenario.}
\label{fig:Simulation_Setup}
\end{figure}

\subsection{Willie's Detection Performance}

To validate the theoretical results related to Willie’s detection performance presented in Lemmas \ref{lemma:DEP}, \ref{lemma:gamma}, \ref{lemma:Min_DEP}, and Theorem \ref{Theorem:K}, we conduct Monte-Carlo simulations and compare the analytical results with the simulated ones.

\begin{figure}[t]
\centering
\includegraphics[width=\figwidth]{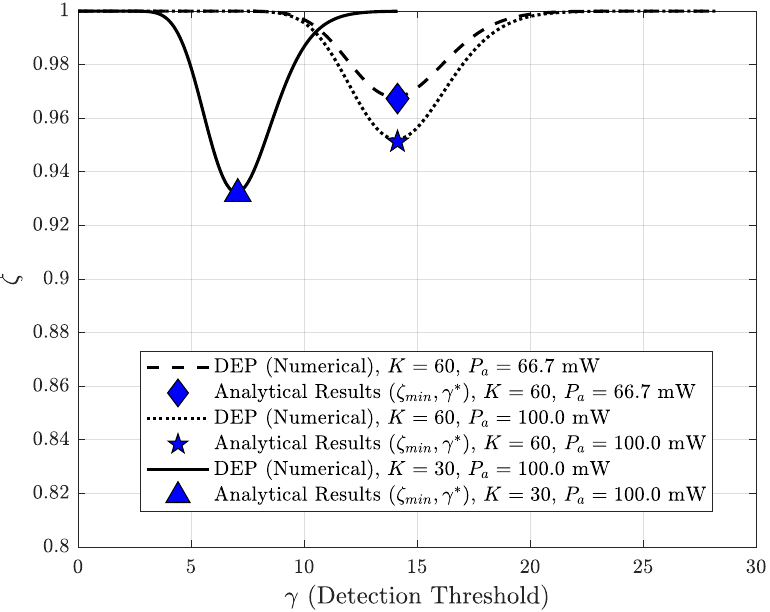}
\caption{DEP $\zeta$ as a function of the detection threshold $\gamma$ for various combinations of $K$ and $P_a$.}
\label{fig:DEP_gamma}
\end{figure}

In Fig. \ref{fig:DEP_gamma}, we present the DEP curves obtained from Monte-Carlo simulations using three different line styles, i.e., dashed, dotted, and solid line, for the combinations of $(K, P_a) = (60, 66.7\,\mathrm{mW})$, $(60, 100\,\mathrm{mW})$, and $(30, 100\,\mathrm{mW})$. The diamond, star, and triangle markers correspond to the theoretical results in Lemmas \ref{lemma:gamma} and \ref{lemma:Min_DEP}. Specifically, the $x$-coordinates represent the optimal detection threshold, $\gamma^*$, derived in Lemma \ref{lemma:gamma} and the $y$-coordinates represent the minimum DEP, $\zeta_{\rm{min}}$, obtained from Lemma~\ref{lemma:Min_DEP}, for the same three $(K, P_a)$ combinations. As observed in Fig. \ref{fig:DEP_gamma}, the theoretical results for $\gamma^{\ast}$ and $\zeta_{\min}$ closely match the Monte-Carlo simulation results for all considered $(K, P_a)$ combinations. This confirms the accuracy of our theoretical expressions in capturing the exact performance. This, in turn, validates the analytical framework established in Lemma \ref{lemma:DEP}, which serves as the fundamental basis for Lemmas \ref{lemma:gamma} and \ref{lemma:Min_DEP}.

Meanwhile, in Fig. \ref{fig:DEP_gamma}, comparing the dotted and solid curves reveals that increasing the number of cooperative users (from $K = 30$ to $K = 60$) leads to a higher minimum DEP. This is attributed to the fact that a larger number of cooperative users generates stronger aggregate interference at Willie, thereby hindering his detection performance. Furthermore, a comparison between the dotted and dashed curves demonstrates that, for a fixed $K$, increasing Alice’s transmit power from $P_a = 66.7\,\mathrm{mW}$ to $100\,\mathrm{mW}$ reduces the DEP across all values of $\gamma$. This occurs because a higher $P_a$ increases the SINR at Willie, which effectively diminishes the impact of the interference and improves his detection capability.

\begin{figure}[t]
\centering
\includegraphics[width=\figwidth]{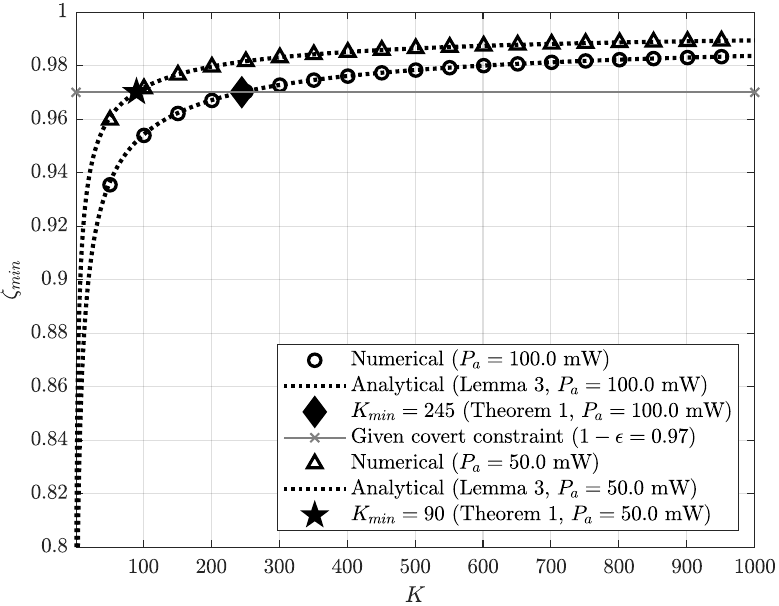}
\caption{Minimum DEP $\zeta_{\min}$ versus the number of cooperative users $K$ for different transmit powers $P_a$.}
\label{fig:MinDEP_vs_K}
\end{figure}

In Fig. \ref{fig:MinDEP_vs_K}, we plot the minimum DEP, $\zeta_{\min}$, as a function of the number of cooperative users $K$ under the on-off scheme in Proposition \ref{prop:power_profile} to validate the theoretical results presented in Lemma \ref{lemma:Min_DEP} and Theorem \ref{Theorem:K}. The open markers, i.e., circle and triangle markers, represent the minimum DEP obtained via Monte-Carlo simulations for $P_a = 100$ and $50~\mathrm{mW}$, respectively. Specifically, the simulation results are determined through an exhaustive search over the detection threshold, $\gamma$. The dotted curves correspond to the theoretical minimum DEP derived in Lemma \ref{lemma:Min_DEP}, showing a precise match with the simulations. Furthermore, the filled star and diamond markers highlight the minimum number of cooperative users required to satisfy the covert constraint $1 - \epsilon = 0.97$, as derived in Theorem \ref{Theorem:K}. The coordinates of the star and diamond markers correspond to the cases of $P_a = 50~\mathrm{mW}$ ($K_{\min} = 90$) and $100~\mathrm{mW}$ ($K_{\min}=245$), respectively. This agreement confirms that Theorem \ref{Theorem:K} serves as a reliable and computationally efficient guideline to determine the necessary number of cooperative users to meet a specified covert constraint. Moreover, a comparison between the star and diamond markers reveals the trade-off between Alice's transmit power and the required cooperative interference resources; specifically, a higher covert transmit power $P_a$ necessitates a larger number of cooperative users $K_{\min}$ to satisfy the same covert constraint, i.e., $1-\epsilon = 0.97$. This is because as $P_a$ increases, Willie's detection capability improves, requiring stronger aggregate interference generated by a large number of users to effectively conceal the signal and satisfy the desired covert constraint.

\subsection{Impact of Covertness Level and Spatial Parameters} \label{Sec:Impactofepsilon}

In this subsection, we investigate the impact of the covert constraint and the spatial heterogeneity of multiple users on the minimum number of cooperative users, $K_{\min}$. In addition, by comparing the theoretical result of this work with that of the existing literature, we validate through comprehensive simulations that the proposed scheme is significantly more efficient than conventional approaches derived from prior work.

\begin{figure}[t]
\centering
\includegraphics[width=\figwidth]{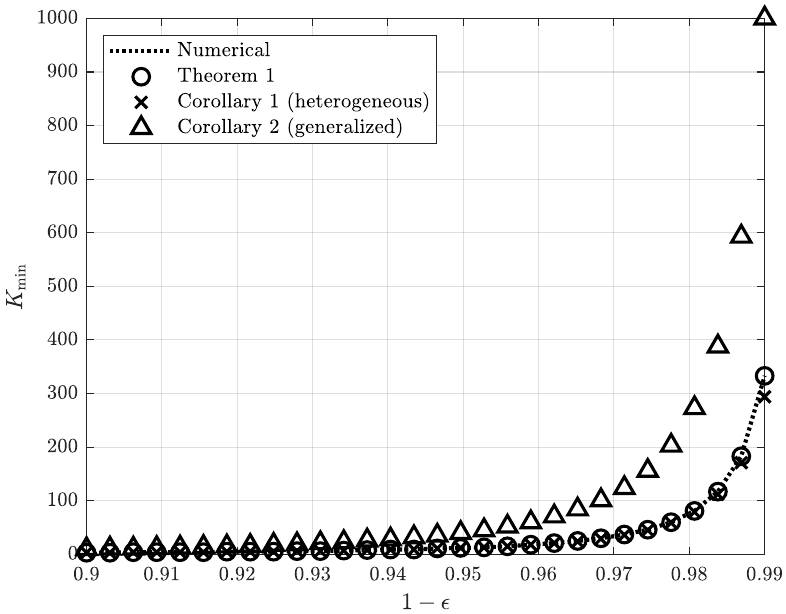}
\caption{Minimum number of cooperative users $K_{\min}$ versus covert constraint $1-\epsilon$ for $P_a = 25~\mathrm{mW}$.}
\label{fig:epsilon_Kmin}
\end{figure}

In Fig. \ref{fig:epsilon_Kmin}, we plot $K_{\min}$ as a function of the covertness parameter, $\epsilon$. The dotted curve represents the values of $K_{\min}$ obtained from Monte-Carlo simulations. The circle, cross, and triangle markers indicate the theoretical results for $K_{\min}$ derived from Theorem \ref{Theorem:K}, Corollaries \ref{Corollary:K_approx}, and \ref{cor:homogeneous}, respectively. Fig. \ref{fig:epsilon_Kmin} illustrates that the theoretical results in Theorem \ref{Theorem:K} closely match the Monte-Carlo simulation results over the entire range of interest, $0.9 \le 1-\epsilon \le 0.99$, which corresponds to the practically relevant region for covert communications. In addition, the close agreement between the cross markers of Corollary \ref{Corollary:K_approx} and the circle markers of Theorem \ref{Theorem:K} confirms that the approximation holds well, indicating that the simulation is conducted in the regime of a sufficiently large number of users $M$. As discussed in Corollary \ref{Corollary:K_approx}, the figure also verifies that as $\epsilon$ decreases, $K_{\min}$ increases according to the inverse-square scaling, i.e., $K_{\min} \propto \mathcal{O}(\epsilon^{-2})$. In particular, for $1-\epsilon = 0.95$, we observe $K_{\min} = 12$, whereas for a stricter constraint of $1-\epsilon = 0.975$, $K_{\min}$ increases to 50. Finally, it is notable that there are significant discrepancies between the exact results (dotted curve) and the homogeneous approximation (triangle markers). In other words, the value of $K_{\min}$ derived under the homogeneous-channel assumption is significantly larger than the exact value. As mentioned earlier, Corollary \ref{cor:homogeneous} generalizes the existing results in \cite{Lee23Multi, Lee24Channel, Yeom24Covert}, which correspond to the special case where $\bar{\lambda}_w = \lambda_{a,w}$. Consequently, this comparison demonstrates that when users are spatially distributed, applying such homogeneous-based analyses leads to a substantial overestimation of $K_{\min}$, thereby failing to provide a proper estimation that reflects the critical characteristics of spatial heterogeneity.

\begin{figure}[t]
\centering
\includegraphics[width=\figwidth]{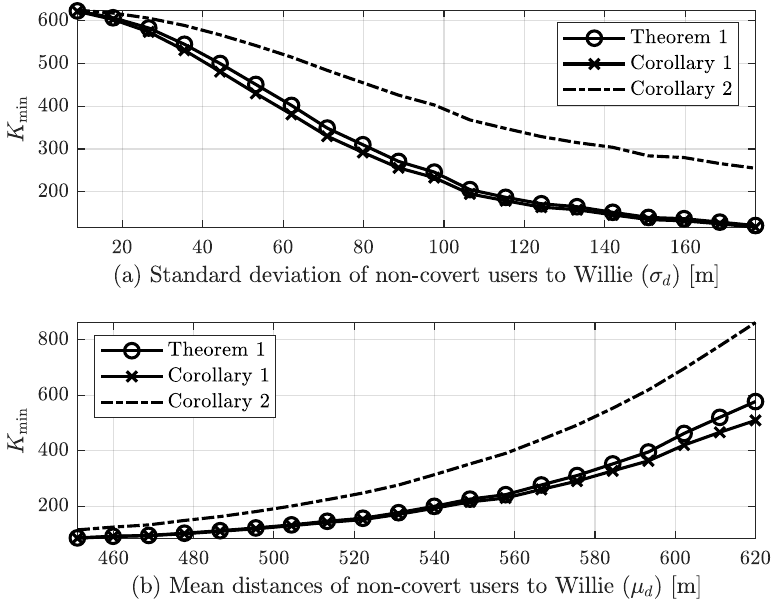}
\caption{The minimum number of cooperative users $K_{\min}$ versus standard deviation for distances of non-covert users to Willie, $\sigma_d$, and mean distances of non-covert users to Willie, $\mu_d$, for $P_a = 100mW$.}
\label{fig:Sensitivity}
\end{figure}

Fig. \ref{fig:Sensitivity} illustrates how the spatial heterogeneity of non-covert users affects $K_{\min}$. In particular, to represent the spatial heterogeneity of non-covert users, Fig. \ref{fig:Sensitivity} (a) uses the standard deviation of their distances to Willie, i.e., $\sigma_d$, and Fig. \ref{fig:Sensitivity} (b) uses the mean distance, i.e., $\mu_d$. In addition, the circle and cross markers indicate the values of $K_{\min}$ derived from Theorem \ref{Theorem:K} and Corollary \ref{Corollary:K_approx}, and the dash-dot curve denotes the $K_{\min}$ obtained from Corollary \ref{cor:homogeneous}. Consistent with the findings in Fig. \ref{fig:epsilon_Kmin}, the baseline design derived from Corollary \ref{cor:homogeneous}, as observed in Figs.\ref{fig:Sensitivity} (a) and (b), leads to an overestimation of the minimum number of cooperative users.

It is also observed in Fig.~\ref{fig:Sensitivity} (a) that as non-covert users become more spatially dispersed, i.e., as $\sigma_d$ increases, the minimum number of cooperative users decreases substantially. It is noteworthy that, as $\sigma_d$ approaches zero, the gap between the proposed heterogeneous analysis in Theorem \ref{Theorem:K} and the homogeneous baseline in Corollary \ref{cor:homogeneous} vanishes. This convergence confirms that our derived model accurately subsumes the conventional homogeneous scenario \cite{Lee23Multi,Lee24Channel, Yeom24Covert} as a special case, validating its consistency and generality. This observation offers a valuable design insight: increasing the spatial dispersion of non-covert users is beneficial for reducing the required minimum number of cooperative users, $K_{\min}$. Since a lower $K_{\min}$ leads to reduced interference at Bob, it ultimately improves the covert rate. Fig.~\ref{fig:Sensitivity} (b) shows that as all users move farther away from Willie, i.e., as $\mu_d$ increases, the required minimum number of cooperative users $K_{\min}$ also increases. This is because a larger $\mu_d$ weakens the aggregate interference power received at Willie, thereby necessitating more cooperative users to satisfy the covert constraint. Combining the results from Figs. \ref{fig:Sensitivity} (a) and (b) yields a comprehensive design guideline: to maximize covert communication performance, non-covert users should be deployed as close to Willie as possible while maintaining a high degree of spatial dispersion.

\subsection{Performance Comparison with Existing Work}

\begin{figure}[t]
\centering
\includegraphics[width=\figwidth]{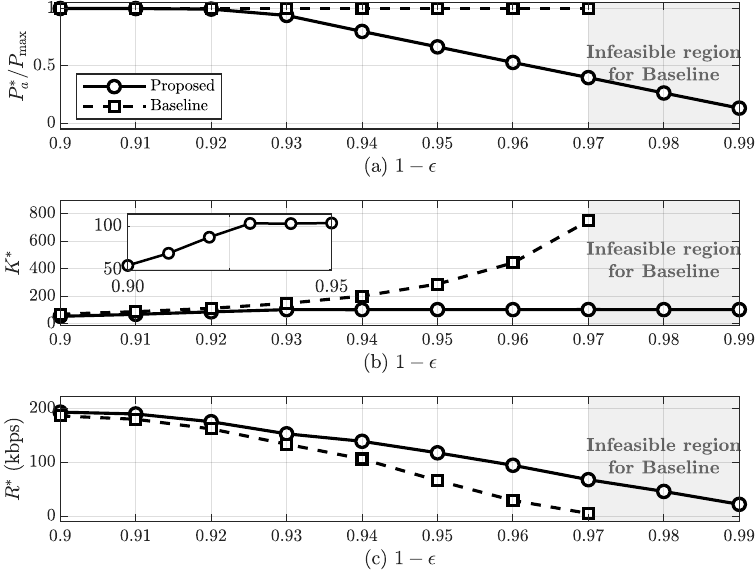}
\caption{Performance comparison between the proposed policy (Proposed) and the Bob-only policy (Baseline) versus the covert constraint $1-\epsilon$.}
\label{fig:Policy_Comparison}
\end{figure}

In this subsection, we compare the covert communication performance of the proposed method with that of the baseline method. In the proposed method, cooperative users are selected based on both the instantaneous channel gains between the users and Bob, and Willie’s large-scale fading coefficients, as described in Proposition \ref{prop:power_profile}. In contrast, the baseline method selects cooperative users solely according to the instantaneous channel gains between the users and Bob \cite{Lee23Multi,Lee24Channel,Yeom24Covert}. In addition, when computing the optimal $K^*$, the proposed method considers the distinct distances between the users and Willie, while the baseline method assumes identical distances from users to Willie and Bob, respectively, despite the actual differences, following the assumption in \cite{Yeom24Covert}. To validate the effectiveness of the proposed method, we compare the optimized results, i.e., $P_a^*$, $K^*$, and $R^*$ with respect to the the stringency of the covert constraint, $1-\epsilon$.

In Figs. \ref{fig:Policy_Comparison} (a) and (b), we present the ratio of the optimized transmit power for the covert message to the maximum transmit power, $P_a^* / P_{\max}$, and the optimal number of cooperative users, $K^*$, respectively. It is observed that in the baseline method, a more stringent covert constraint, i.e., a smaller $\epsilon$, is met primarily by increasing the number of cooperative users rather than decreasing Alice's transmit power. Meanwhile, the proposed method adopts a similar strategy when the covertness requirement is moderate, i.e., $1-\epsilon < 0.93$. However, when the requirement becomes strict, i.e., $1-\epsilon \ge 0.93$, the proposed method shifts its strategy by decreasing the transmit power for the covert message instead of further increasing the number of cooperative users. This strategic shift occurs because, in a spatially heterogeneous network, the interference effectiveness varies significantly among users. The proposed method exploits this spatial heterogeneity by selectively utilizing only those users who effectively interfere with Willie without excessively interfering with Bob. Beyond a certain point, adding more inefficient users yields diminishing returns, making power reduction the optimal choice. In contrast, since the baseline method assumes homogeneous channel conditions, it incorrectly perceives all users as equally effective interferers, leading to a continuous and inefficient increase in the number of cooperative users.

Fig. \ref{fig:Policy_Comparison} (c) demonstrates that the proposed method improves the covert rate compared with the baseline method. In particular, as the covert constraint becomes more stringent, the performance gap between the proposed and baseline methods grows larger. This indicates that incorporating the distinct spatial locations of users into the optimization process substantially enhances covert performance, especially under stringent covert constraints, thereby highlighting the necessity of the proposed approach. Lastly, it is observed in Figs. \ref{fig:Policy_Comparison} (a), (b), and (c) that when a stringent covert constraint is required, i.e., $1-\epsilon \ge 0.97$, the baseline method fails to satisfy the covert constraint. In contrast, the proposed method is able to meet such a stringent covert constraint requirement. This implies that accounting for the distinct spatial locations of users broadens the range of achievable covert constraint stringency.

\section{Conclusion}\label{Sec:conclusion}

This paper investigated uplink multi-user covert communication in which spatially distributed users assist the covert transmission. We mathematically proved that, under the considered system model, the on--off scheme is the optimal policy for designing the interference power transmitted by distributed users. With the optimal on--off scheme, closed-form expressions were derived for the minimum number of cooperative users and the activation threshold required to satisfy a given covert constraint. Based on these expressions, we conducted a comprehensive analysis of the covert communication performance with respect to system parameters. In addition, we proposed an intelligent piecewise search algorithm to maximize the achievable rate under both power and covert constraints, jointly optimizing the transmit power of the covert message and the interferences transmitted by spatially heterogeneous users. Crucially, we demonstrated for the first time that the spatial heterogeneity of distributed users acts as a fundamental advantage, significantly enhancing covert communication performance compared to conventional homogeneous scenarios.

\appendices

\section{Proof of Proposition \ref{prop:power_profile}} \label{Appendix:power_profile}

Given a fixed $P_a$, the optimization problem in~\eqref{eq:original_opt} can be reformulated as
\begin{subequations}\label{eq:powerprofile_Opt}
\begin{align}
  \bar{P}^* & = \argmax_{\bar{P}} \frac{P_a |h^a_{a,b}|^2}{\sum_{m=1}^M P_m|h^a_{m,b}|^2 + \sigma_b^2} \\
  & \text{s.t.} ~~ \eqref{eq:CovertConst} ~ \text{and} ~ \eqref{eq:PowConst_Noncovert}.\nonumber
\end{align}
\end{subequations}

First, we analyze the minimum DEP associated with the covert constraint in \eqref{eq:CovertConst}. To derive the DEP, we employ a moment-matching based approximation technique, where the sum of Gamma-distributed random variables is approximated as another Gamma distribution with the same first and second moments. The feasibility and tightness of this approximation scheme have been extensively analyzed and validated in \cite{Mathai93AHandbook,Moschopoulos85TheDistribution,Covo14AnovelSingleGammaApprox}. Thus, $T_w$ in \eqref{eq:test_statistic} can be reformulated to 
\begin{equation} \label{eq:pdfTest}
T_w - \sigma_w^2 \sim \begin{cases}
    \Gamma\!\left( \frac{\psi^2}{\Psi},\,\frac{\Psi}{\psi} \right), & \mathcal{H}_0, \\[3pt]
    \Gamma\!\left( \frac{(\psi + \lambda_{a,w}P_a)^2}{\Psi + (\lambda_{a,w} P_a)^2},\,\frac{\Psi + (\lambda_{a,w} P_a)^2}{\psi + \lambda_{a,w}P_a} \right), & \mathcal{H}_1,
  \end{cases}
\end{equation}
where $\psi = \sum_{m=1}^M \lambda_{m,w}P_m$ and $\Psi = \sum_{m=1}^M \lambda_{m,w}^2P_m^2$.
For sufficiently large $M$, the Gamma distribution converges to a normal distribution \cite{Papoulis02Probability}, yielding
\begin{equation} \label{eq:ProbDist_Tw}
    T_w - \sigma_w^2 \sim \begin{cases}
      \mathcal{N}\!\left(\psi,\,\Psi \right), & \mathcal{H}_0, \\[3pt]
      \mathcal{N}\!\left(\psi + \lambda_{a,w}P_a,\, \Psi + \lambda^2_{a,w}P^2_a \right), & \mathcal{H}_1.
\end{cases}
\end{equation}
Based on the probabilistic distribution of \eqref{eq:ProbDist_Tw} and the definition of the DEP in \eqref{eq:DEP}, the DEP can be calculated as
\begin{align}
\zeta &= \Pr(T_w > \gamma \mid \mathcal{H}_0) + \Pr(T_w \leq \gamma \mid \mathcal{H}_1) \nonumber \\
&= \underbrace{Q\!\left(\frac{\gamma-\psi}{\sqrt{\Psi}}\right)}_{(a)} 
   + \underbrace{Q\!\left(\frac{(\psi + \lambda_{a,w}P_a) - \gamma}{\sqrt{\Psi + \lambda_{a,w}^2P_a^2}}\right)}_{(b)}, \label{eq:DEP_Gaussian}
\end{align}
where $Q(\cdot)$ denotes the standard Gaussian $Q$-function. Assuming that the false alarm probability is set to $\alpha$, i.e., $P_{\rm FA} = \alpha$, the detection threshold $\gamma$ can be obtained from result (a) in \eqref{eq:DEP_Gaussian} as
\begin{equation}
    \gamma = \sqrt{\Psi}\,Q^{-1}(\alpha) + \psi, \label{eq:FAwithAlpha}
\end{equation}
and with the result (b) in \eqref{eq:DEP_Gaussian} and \eqref{eq:FAwithAlpha}, the corresponding miss detection probability is given by
\begin{equation}
    P_{\rm MD} = Q\!\left(\frac{\lambda_{a,w}P_a - \sqrt{\Psi}\,Q^{-1}(\alpha)}{\sqrt{\Psi + \lambda_{a,w}^2P_a^2}}\right).
\end{equation}
Consequently, the DEP, as a function of $\alpha$, can be expressed as
\begin{equation}\label{eq:DEP_Pi}
\zeta = \alpha + Q\!\left(\frac{\lambda_{a,w}P_a - \sqrt{\Psi}\,Q^{-1}(\alpha)}{\sqrt{\Psi + \lambda_{a,w}^2P_a^2}}\right).
\end{equation}
It is noteworthy that $\zeta$ in \eqref{eq:DEP_Pi} is a monotonically increasing function of $\Psi$ for any $0 \leq \alpha \leq 1$. Moreover, the minimum detection error probability $\zeta_{\min}$ also increases monotonically with $\Psi$. Suppose that for two values $\Psi < \hat\Psi$, there exist optimal thresholds $\alpha^*$ and $\hat\alpha$ minimizing $\zeta(\alpha,\Psi)$ and $\zeta(\alpha,\hat\Psi)$, respectively. Then, by definition,
\[
\zeta_{\min}(\Psi) = \zeta(\alpha^*,\Psi) \leq \zeta(\hat\alpha,\Psi),
\]
and because $\zeta$ increases with $\Psi$ for any fixed $\alpha$,
\[
\zeta(\hat\alpha,\Psi) < \zeta(\hat\alpha,\hat\Psi) = \zeta_{\min}(\hat\Psi).
\]
Therefore, $\zeta_{\min}(\Psi)$ is strictly increasing with respect to $\Psi$. By applying Sedrakyan’s Lemma~\cite{Sedrakyan97AboutTheApplications}, we obtain
\begin{equation} \label{eq:Cauchy}
  \frac{\Psi}{M^2}
  = \frac{1}{M}\sum_{m=1}^M \!\left(\frac{\lambda_{m,w}P_m}{M}\right)^2
  \ge \frac{1}{M^3}\!\left(\sum_{m=1}^M \lambda_{m,w}P_m\right)^2,
\end{equation}
where the inequality becomes tight as $M$ grows large. Hence, for sufficiently large $M$, we have $\Psi \approx \psi^2/M$, and $\zeta_{\min}$ becomes a monotonically increasing function of $\psi^2/M$. Therefore, the covert constraint can be equivalently expressed as
\begin{equation}\label{eq:convert_DEP_Pi}
  \psi = \sum_{m=1}^M \lambda_{m,w}P_m \geq \delta,
\end{equation}
where $\delta \triangleq \sqrt{M}\,\zeta_{\min}^{-1}(1-\epsilon)$.

The optimization problem in \eqref{eq:powerprofile_Opt} can be reformulated as
\begin{subequations}\label{eq:powerprofile_Opt_reformul}
\begin{align}
  \bar{P}^* = \arg\min_{\bar{P}} & \sum_{m=1}^M P_m|h^a_{m,b}|^2 + \sigma_b^2,  \label{eq:powerprofile_Opt_reformul_obj}\\
  & \text{s.t.} ~~ \sum_{m=1}^M \lambda_{m,w}P_m \ge \delta ~~ \text{and} ~~ \eqref{eq:PowConst_Noncovert}     \label{eq:CovertConst_IntfVer}
\end{align}
\end{subequations}
Now, to solve the optimization problem in \eqref{eq:powerprofile_Opt_reformul}, we first derive the Karush–Kuhn–Tucker (KKT) conditions with Lagrange multipliers, i.e., $\{ \phi_m, \omega_m, \nu | m \in \{ 1,2,\ldots, M\} \}$, as
\begin{itemize}
  \item \textbf{Stationarity:} $|h^a_{m,b}|^2 - \phi_m^* + \omega_m^* + \nu^*\lambda_{m,w} = 0$, $\forall m$.
  \item \textbf{Complementary slackness:} $\phi_m^* P_m^* = 0$, $\omega_m^*(P_m^* - P_{\max}) = 0$, and $\nu^* (\sum_{m=1}^{M} \lambda_{m,w}P_m - \delta) = 0$.
  \item \textbf{Feasibility:} $0 \leq P_m^* \leq P_{\max}$, $\sum_{m=1}^M \lambda_{m,w}P_m^* \ge \delta$, and $\phi_m^*, \omega_m^*, \nu^* \ge 0$.
\end{itemize}
The feasible solutions, which correspond to the combinations of optimization variables and Lagrange multipliers satisfying the KKT conditions, are summarized in Table \ref{Tbl:KKT_Conditions}.
\begin{table}[t]
\caption{Feasible solutions under the KKT conditions (NS: no solution).}
\label{Tbl:KKT_Conditions}
\centering
\renewcommand{\arraystretch}{1.1}
\setlength{\tabcolsep}{3pt}
\begin{tabular}{|c|c|c|c|}
\hline
 & $\tfrac{|h^a_{m,b}|^2}{\lambda_{m,w}} > \nu^*$ 
 & $\tfrac{|h^a_{m,b}|^2}{\lambda_{m,w}} = \nu^*$ 
 & $\tfrac{|h^a_{m,b}|^2}{\lambda_{m,w}} < \nu^*$ \\ \hline
$\omega_m^* = 0$ 
 & $P_m^* = 0$ 
 & $\lambda_{m,w}P_m^* = \delta - \!\!\sum_{m \in \mathcal{C}}\!\! P_j^*$
 & NS \\ \hline
$\omega_m^* > 0$ 
 & NS & NS & $P_m^* = P_{\max}$ \\ \hline
\end{tabular}
\end{table}
Thus, the optimal transmit power for each non-covert user is expressed as
\begin{equation}
    P_m^* =
\begin{cases}
      P_{\max}, & \frac{|h^a_{m,b}|^2}{\lambda_{m,w}} < \tau, \\[3pt]
      0, & \frac{|h^a_{m,b}|^2}{\lambda_{m,w}} > \tau, \\[3pt]
      \displaystyle \frac{\delta - \sum_{ m \in \mathcal{C} } P_{\max}}{\lambda_{m,w}}, & \frac{|h^a_{m,b}|^2}{\lambda_{m,w}} = \tau,
\end{cases} \label{eq:OptPowProfile}
\end{equation}
where $\mathcal{C}$ denotes the set of selected cooperative users, i.e., $\mathcal{C} = \{m': |h^a_{m',b}|^2 / \lambda_{m',w} < \tau \}$ and $\tau = \nu^*$. Note that if $U_m$ is a non-covert user satisfying $|h^a_{m,b}|^2 / \lambda_{m,w} = \tau$,  the user cannot determine the transmit power from \eqref{eq:OptPowProfile} since $U_m$ cannot have the knowledge of the set $\mathcal{C}$, i.e., the channels between the other users and Bob. Consequently, the optimal power profile is given by the following on–off scheme, as expressed in \eqref{eq:optimal_Pm}.
\begin{align}\label{eq:optimal_Pm}
P_m^* =
\begin{cases}
P_{\max}, & \frac{|h^a_{m,b}|^2}{\lambda_{m,w}} \le \tau, \\[3pt]
0, & \frac{|h^a_{m,b}|^2}{\lambda_{m,w}} > \tau.
\end{cases}
\end{align}
\qed

\section{Proof of Lemma \ref{lemma:DEP}}\label{Appendix:DEP}

By leveraging the moment-matching based Gamma approximation technique introduced in \cite{Mathai93AHandbook, Moschopoulos85TheDistribution, Covo14AnovelSingleGammaApprox}, the PDF of $T_w - \sigma_w^2$ in \eqref{eq:test_statistic_RV} can be approximated as
\begin{equation} \label{eq:test_statistic_RV_GammaApprox}
   T_w - \sigma_w^2  \\[-2pt]
      \sim 
     \begin{cases}
      \displaystyle \Gamma\!\left( \frac{\theta^2}{\Theta},\,\frac{\Theta}{\theta} \right), 
      \hfill \mathcal{H}_0, \\
      \displaystyle \Gamma\!\left( 
        \frac{\theta^2}{\Theta},
        \frac{\Theta }{\theta}\right) \ast \Gamma\left(1, \lambda_{a,w}P_a
      \right), 
      \hfill \mathcal{H}_1,
     \end{cases} 
\end{equation}
where $\theta = \sum_{i=1}^K \lambda_{m_i,w}P_{\max}$ and $\Theta = \sum_{i=1}^K \lambda_{m_i,w}^2P_{\max}^2$. Recall that a Gamma distribution converges to a Gaussian distribution as its shape parameter increases \cite{Ross19Probability,Casella02StatInference}. In our scenario, the sufficiently large number of cooperative users, $K$, results in a large shape parameter in \eqref{eq:test_statistic_RV_GammaApprox}. Consequently, the equation of \eqref{eq:test_statistic_RV_GammaApprox} can be approximated to
\begin{equation} \label{eq:test_statistic_RV_CLTApprox}
   T_w - \sigma_w^2  \\[-2pt]
      \sim 
     \begin{cases}
      \displaystyle \mathcal{N} \left( \theta, \Theta \right), 
      \hfill \mathcal{H}_0, \\
      \displaystyle \mathcal{N} \left( \theta, \Theta \right) \ast \Gamma\left(1, \lambda_{a,w}P_a
      \right), 
      \hfill \mathcal{H}_1.
     \end{cases} 
\end{equation}
Practically, the set of selected users is unknown to Willie. Let $\mathcal{S}\subset\{1,\dots,M\}$ with $|\mathcal{S}|=K$ denote the (unknown) selection, drawn uniformly at random without replacement. Conditioned on $\mathcal{S}$, the equation of \eqref{eq:test_statistic_RV_CLTApprox} can be reformulated by 
\begin{equation} \label{eq:test_statistic_RV_CLTApprox_fixedSet}
   T_w - \sigma_w^2 \mid \mathcal{S} \\[-2pt]
      \sim 
     \begin{cases}
      \displaystyle \mathcal{N} \left( \theta_\mathcal{S}, \Theta_\mathcal{S} \right), 
      \hfill \mathcal{H}_0, \\
      \displaystyle \mathcal{N} \left( \theta_\mathcal{S}, \Theta_\mathcal{S} \right) \ast \Gamma\left(1, \lambda_{a,w}P_a
      \right), 
      \hfill \mathcal{H}_1,
     \end{cases} 
\end{equation}
where $\theta_{\mathcal{S}}=\sum_{m\in\mathcal{S}}\lambda_{m,w}P_{\max}$ and $\Theta_{\mathcal{S}}=\sum_{m\in\mathcal{S}}\lambda_{m,w}^{2}P_{\max}^{2}$. 
As $\mathcal{S}$ is randomly drawn and unknown to Willie, the overall statistics of $T_w-\sigma_w^2$ must incorporate this selection induced randomness. Equivalently, the mean and variance must be computed by averaging the conditional moments with respect to the random user selection set $\mathcal{S}$.

By the law of total expectation and variance~\cite{Ross19Probability,Casella02StatInference},
the overall mean, $\Xi$ and variance, $\Sigma$ of $T_w-\sigma_w^2$ are obtained by averaging the conditional mean $\theta_{\mathcal{S}}$ and variance $\Theta_{\mathcal{S}}$ over the random user selection $\mathcal{S}$,
yielding
\begin{align}
\Xi 
&\triangleq \mathbb{E}_{\mathcal{S}}[\theta_{\mathcal{S}}]
= \frac{K}{M}\sum_{i=1}^{M}\lambda_{i,w}P_{\max},
\label{eq:Xi_marg}
\end{align}
and
\begin{align}\label{eq:Sigma_marg}
\Sigma 
&\triangleq \mathbb{E}_{\mathcal{S}}[\Theta_{\mathcal{S}}]
+ \mathrm{Var}_{\mathcal{S}}[\theta_{\mathcal{S}}] \nonumber\\[4pt]
&= P_{\max}^2 \Bigg[
   \frac{K}{M}\sum_{i=1}^{M}\lambda_{i,w}^{2}
   + \frac{K(M-K)}{M-1} \nonumber \\[-2pt] 
   &\hspace{5em}\times \frac{1}{M}\sum_{i=1}^{M}
     \left(
       \lambda_{i,w}
       - \frac{1}{M}\sum_{j=1}^{M}\lambda_{j,w}
     \right)^{2}
    \Bigg],
\end{align}
where $\mathbb{E}_{\mathcal{S}}[\Theta_{\mathcal{S}}]$ is the expected conditional variance of the interference, 
and $\mathrm{Var}_{\mathcal{S}}(\theta_{\mathcal{S}})$ is the variance of the conditional mean due to random user selection. The latter term, $\mathrm{Var}_{\mathcal{S}}(\theta_{\mathcal{S}})$, incorporates the finite population correction (FPC) \cite{Cochran77Sampling,Lohr19Sampling} since the users are sampled without replacement.

Accordingly, the false-alarm probability, $P_{FA}$, is defined as
\begin{equation}\label{eq:PFA_Qform}
P_{\mathrm{FA}}
= \Pr\!\left(T_w - \sigma_w^2 > \hat{\gamma} \,\middle|\, \mathcal{H}_0 \right)
\approx Q\!\left(\frac{\hat{\gamma} - \Xi}{\sqrt{\Sigma}}\right),
\end{equation}
where $\hat{\gamma}\triangleq \gamma-\sigma_w^2$. In addition, the miss-detection probability at Willie can be derived by
\begin{align}\label{eq:PMD_Qform}
P_{\mathrm{MD}}
&= \Pr\!\left(T_w-\sigma_w^2 \le \hat{\gamma}\,;\,\mathcal{H}_1\right) = \Pr(X+Y \le \hat{\gamma}) \nonumber\\
& = \int_{-\infty}^{\infty}
    \frac{1}{\sqrt{2\pi\Sigma}}
    \exp\!\left(-\frac{(y - \Xi)^2}{2\Sigma}\right)
    \gamma\!\left(1, \frac{\hat{\gamma} - y}{\Delta}\right)
\, dy \nonumber\\
&= Q\!\left(-\frac{\hat{\gamma}-\Xi}{\sqrt{\Sigma}}\right)
 - \exp\!\left(-\frac{2\Delta(\hat{\gamma}-\Xi)-\Sigma}{2\Delta^2}\right) \nonumber \\
  & \hspace{5em} \times Q\!\left(-\frac{\Delta(\hat{\gamma}-\Xi)-\Sigma}{\sqrt{\Sigma}\,\Delta}\right),
\end{align}
where \(X\sim\Gamma(1,\Delta)\) with \(\Delta\triangleq P_a\lambda_{a,w}\) and \(Y\sim\mathcal{N}(\Xi,\Sigma)\). Combining \eqref{eq:PMD_Qform} with \eqref{eq:PFA_Qform}
gives
\begin{align}
\zeta (\gamma)
&= P_{\mathrm{FA}} + P_{\mathrm{MD}} \nonumber\\
&= 1 - \exp\!\left(-\frac{2\Delta(\hat{\gamma}-\Xi)-\Sigma}{2\Delta^2}\right)
       \Phi\!\left(\frac{\hat{\gamma}-\Xi-\Sigma/\Delta}{\sqrt{\Sigma}}\right) \nonumber\\
&= 1 - \exp\!\left(-\frac{2\Delta(\hat{\gamma}-\Xi)-\Sigma}{2\Delta^2}\right)
       Q\!\left(-\,\frac{\Delta(\hat{\gamma}-\Xi)-\Sigma}{\sqrt{\Sigma}\,\Delta}\right), \nonumber
\end{align}
which matches \eqref{eq:DEP_general}. 
\qed

\section{Proof of Theorem \ref{Theorem:K}} \label{Append:ProofMinK}

From~\eqref{eq:Opt_DEP}, the condition $\zeta_{\min}=1-\epsilon$ can be written as
\begin{equation}
1-\frac{1}{\sqrt{\pi}\!\left(\sqrt{\tfrac{\Sigma}{2\Delta^2}}+\sqrt{\tfrac{\Sigma}{2\Delta^2}+\tfrac{4}{\pi}}\right)}=1-\epsilon,
\end{equation}
which is equivalent to
\begin{equation}\label{eq:Sigma_condition_append}
\frac{\Sigma}{2\Delta^2}\ = \left(\frac{1-4\epsilon^2}{2\epsilon\sqrt{\pi}}\right)^{\!2}.
\end{equation}
Let $c_\epsilon \triangleq \frac{1/\epsilon^2 - 8 + 16\epsilon^2}{2\pi}$ and recall $\Delta=P_a\lambda_{a,w}$.
By substituting the marginal second-moment term $\Sigma$ from~\eqref{eq:Sigma_marg} into~\eqref{eq:Sigma_condition_append}, we obtain
\begin{multline} \label{eq:K_condition_start}
\frac{1}{(P_a\lambda_{a,w})^2}
P_{\max}^2
\Bigg[
    \frac{K}{M}\sum_{m=1}^{M}\lambda_{m,w}^{2} \\
    + \frac{K(M-K)}{M-1}\cdot \frac{1}{M}\sum_{i=1}^{M}
\Big(\lambda_{i,w}-\tfrac{1}{M}\!\sum_{j=1}^{M}\lambda_{j,w}\Big)^{2}
\Bigg]
 = c_\epsilon.
\end{multline}
Define the empirical mean-square and variance
\[
E \triangleq \mathbb{E}[\lambda_w^2]\!=\!\frac{1}{M}\sum_{i=1}^{M}\lambda_{i,w}^{2}
\]
and
\[
V \triangleq \mathrm{Var}[\lambda_w]\!=\!\frac{1}{M}\sum_{i=1}^{M}\!
\Big(\lambda_{i,w}-\tfrac{1}{M}\sum_{j=1}^{M}\lambda_{j,w}\Big)^{2}.
\]
Then~\eqref{eq:K_condition_start} can be rearranged as
\begin{equation}\label{eq:K_condition_general}
\frac{P_{\max}^2 K}{P_a^2 \lambda_{a,w}^2}\,\Big(E+\tfrac{M-K}{M-1}\,V\Big)\ = c_\epsilon.
\end{equation}
For large $M$, $\tfrac{M-K}{M-1} \approx \tfrac{M-K}{M} = 1-\tfrac{K}{M}$, so that
\[
E+\tfrac{M-K}{M-1}V \ \approx\ E+V - \tfrac{K}{M}V.
\]
Introduce the constant, $C \triangleq P_{\max}^2 / (P_a^2\lambda_{a,w}^2c_\epsilon)$,
and rewrite \eqref{eq:K_condition_start} with equality at the threshold as
\[
CK\Big(E+V-\tfrac{K}{M}V\Big) = 1.
\]
Solving via the quadratic formula and simplifying gives
\[
K =\frac{M}{2V}\left((E+V)\ \pm\ \sqrt{(E+V)^2-\frac{4V}{CM}}\right).
\]
Since we seek the minimum $K$, the physically relevant solution is the smaller root,
\[
K_1 = \frac{M}{2V}\left((E+V)-\sqrt{(E+V)^2-\frac{4V}{CM}}\right).
\]
Consequently, considering the integer constraint on $K$, the minimum number of cooperative users is determined as $K_{\min} = \lceil K_1 \rceil$. This concludes the proof.
\qed

\bibliographystyle{IEEEtran}
\bibliography{ref_covert}

\end{document}